\documentclass[letterpaper, 10 pt, conference]{ieeeconf}  
\IEEEoverridecommandlockouts                              
\usepackage[utf8]{inputenc}
\usepackage{amsmath}
\usepackage{amsthm}
\usepackage{mathrsfs}
\usepackage{amssymb}
\usepackage{color}
\usepackage{cite} 
\usepackage[font=footnotesize]{caption}
\usepackage{comment}
\usepackage{float}
\usepackage{graphicx}
\usepackage{float}
\usepackage{lscape}
\usepackage{bm}
\usepackage{booktabs}
\usepackage{longtable,tabularx,ragged2e}
\usepackage[super]{nth}
\usepackage{multicol}
\usepackage{booktabs}
\usepackage{lscape}
\usepackage{nicefrac}
\usepackage[hidelinks]{hyperref}
\usepackage{mathtools}
\usepackage{siunitx}

\newtheorem{lma}{Lemma}
\newtheorem{theo}{Theorem}
\newtheorem{prop}{Proposition}
\newtheorem{cor}{Corollary}
\newtheorem{defn}{Definition}

\newtheorem{lem}{Lemma}

\DeclareMathOperator{\rank}{rank}
\DeclareMathOperator{\tr}{tr}

\newcommand{\eigmin}{\boldsymbol{\lambda_{\rm min}}}

\newcommand{\rmT}{{\rm T}}

\newcommand{\BBN}{{\mathbb N}}

\newcommand{\BBR}{{\mathbb R}}

\newcommand{\SF}{{\mathcal F}}
\newcommand{\SG}{{\mathcal G}}

\newcommand{\SU}{{\mathcal U}}

\newcommand{\SY}{{\mathcal Y}}

\allowdisplaybreaks

\newcounter{example}
\newenvironment{example}[1][]{\refstepcounter{example}\par\medskip
   \noindent \textbf{\indent Example~\theexample. #1} \rmfamily}{\medskip}
   
\DeclareMathOperator*{\argmin}{arg\,min}

\newtheorem{lemalt}{Lemma}[lma]

\newenvironment{lema}[1]{
  \renewcommand\thelemalt{#1}
  \lemalt
}{\endlemalt}

\title{Convergence of Recursive Least Squares Based Input/Output System Identification with Model Order Mismatch}

\author{\large Brian Lai and Dennis S. Bernstein 
\thanks{Brian Lai and Dennis S. Bernstein are with the Department of Aerospace Engineering, University of Michigan, Ann Arbor, MI, USA. {\tt \{brianlai,  dsbaero\}@umich.edu}. This work was supported by the NSF Graduate Research Fellowship under Grant No. DGE 1841052.}
}

\begin{document}
\maketitle

\begin{abstract}
    Discrete-time input/output models, also called infinite impulse response (IIR) models or autoregressive moving average (ARMA) models, are useful for online identification as they can be efficiently updated using recursive least squares (RLS) as new data is collected. 
    Several works have studied the convergence of the input/output model coefficients identified using RLS under the assumption that the order of the identified model is the same as that of the true system.
    However, the case of model order mismatch is not as well addressed.
    This work begins by introducing the notion of \textit{equivalence} of input/output models of different orders.
    Next, this work analyzes online identification of input/output models in the case where the order of the identified model is higher than that of the true system.
    It is shown that, given persistently exciting data, the higher-order identified model converges to the model equivalent to the true system that minimizes the regularization term of RLS.
\end{abstract}

\section{Introduction}

Least squares based methods are fundamental and widely used in identification, signal processing, and control \cite{aastrom1995adaptive,ljung1983theory}.
One useful application is online identification of a linear model for adaptive control \cite[ch. 3]{goodwin2014adaptive}.
A particular model structure that lends itself to online identification is a discrete-time input/output model \cite{willems1986time}, also called an infinite-impulse-response (IIR) model \cite{cuevas2014comparison} or autoregressive moving-average (ARMA) model \cite[p. 32]{goodwin2014adaptive}.
An advantage of discrete-time input/output models is they can be efficiently updated in real time as new data is collected \cite[ch. 3.2]{goodwin2014adaptive}.
Note that in the single-input-single-output (SISO) case, the coefficients of a discrete-time input/output model directly give the discrete-time transfer function.
However, the multi-input-multi-output (MIMO) case is considerably more complex \cite{willems1986time,willems1997introduction} and converting between input/output models and other linear model structures is nontrivial \cite{shieh1982transformations,al2002identification}.
As such, it may be beneficial to directly study the online identification of MIMO input/output models without converting to another model structure.

This paper focuses on the online identification of input/output models using recursive least squares (RLS), which has been used in adaptive model predictive control \cite{nguyen2021predictive} and retrospective cost adaptive control \cite{islam2021data} with various applications \cite{lai2023data,mohseni2022adaptive,farahmandi2024predictive}.
While related works have discussed convergence of model coefficients when identifying input/output models using RLS \cite[ch. 3.4]{goodwin2014adaptive}, \cite{ding2009multiinnovation}, these results assume the order of the true model is known.
A natural question is whether similar guarantees can be made if the order of the identified model and true system do not match.
This work will show that if the order of the identified model is higher than order of the true model, then the regressor of RLS is not persistently exciting, and hence standard convergence guarantees of RLS \cite{bruce2021necessary,lai2021regularization,goodwin2014adaptive} do not apply.
The main contribution of this work is developing new analysis of the case where the order of the identified input/output model is higher than that of the true input/output system.

This paper is organized as follows. 
Section \ref{sec: IO Modeling} introduces discrete-time input/output models and a useful output transition equation.
Next, Section \ref{sec: equivalence and reduciblity} introduces the notion of \textit{equivalence} of input/output models as models giving the same outputs under the same inputs and initial conditions. 
It is shown that a necessary and sufficient condition for equivalence can be written as a linear equation of the model coefficients.
This section also introduces the notion of \textit{reducibility} as the existence of a lower-order, equivalent input/output model.
Finally, section \ref{sec: online ID of IO Models} discusses the online identification of input/output models using RLS. 
It is first shown that, in the case where the order of the identified model is the same as that of the true system, persistent excitation conditions guarantee global asymptotic stability of the coefficient estimation error.
Next, in the case where the order of the identified model is higher than that of the true system, conditions are given under which the identified model converges to the higher-order model equivalent to the true system that minimizes the regularization term of RLS.

\section{Input/Output Modeling}
\label{sec: IO Modeling}

Let $k_0 \in \BBN$ be the initial time step and let $n \ge 0$ be the model order. Consider the input/output model where, for all $k \ge k_0$, $u_k \in \BBR^m$ is the input, $y_{k_0},\hdots,y_{k_0+n-1} \in \BBR^{p}$ are the initial conditions and, for all $k \ge k_0$, the output $y_{k+n} \in \BBR^p$ is given by
\begin{align}
\label{eqn: IO Model}
    y_{k+n} = -\sum_{i=1}^n F_i y_{k+n-i} + \sum_{i=0}^n G_i u_{k+n-i},
\end{align}
where $F_1,\hdots,F_n \in \BBR^{p \times p}$, and $G_0,\hdots,G_n \in \BBR^{p \times m}$ are the input/output model coefficients.
It follows that, for all $k \ge k_0$,
\begin{align}
\label{eqn: IO matrix form}
    y_{k+n} = -\SF_n \SY_{k,n} + \SG_n \SU_{k,n},
\end{align}
where $\SF_n \in \BBR^{p \times pn}$, $\SY_{k,n} \in \BBR^{pn}$, $\SG_n \in \BBR^{p \times m(n+1)}$, and $\SU_{k,n} \in \BBR^{m(n+1)}$ are defined as
\begin{align}
    \SF_n &\triangleq  \begin{bmatrix}
        F_1 & \cdots & F_n
    \end{bmatrix}, 
    \quad 
    \SY_{k,n} \triangleq \begin{bmatrix}
        y_{k+n-1} \\ \vdots \\ y_k
    \end{bmatrix},
    \label{eqn: SF and SY definition}
    \\
    \SG_n &\triangleq \begin{bmatrix}
        G_0 & \cdots & G_n
    \end{bmatrix},
    \quad 
    \SU_{k,n} \triangleq  \begin{bmatrix}
        u_{k+n} \\ \vdots \\ u_k
    \end{bmatrix}.
    \label{eqn: SG and SU definition}
\end{align}
%
%
Proposition \ref{prop: IO j step update} shows that, for all $k \ge k_0$ and $j \ge 0$, $y_{k+n+j}$ can be written as a linear combination of the inputs $u_{k},\hdots,u_{k+n+j-1}$ and the $n$ outputs $y_{k},\hdots,y_{k+n-1}$.
We call this the output transition equation.

\begin{prop}
\label{prop: IO j step update}
Let $k_0 \in \BBN$. For all $k \ge k_0$, let $u_k \in \BBR^m$, let $y_{k_0},\hdots,y_{k_0+n-1} \in \BBR^p$, and, for all $k \ge k_0$, let $y_{k+n} \in \BBR^p$ be given by \eqref{eqn: IO Model}. Then, for all $k \ge k_0$ and $j \ge 0$,
\begin{align}
\label{eqn: IO j step update}
    y_{k+n+j} = -\SF_{n,j} \SY_{k,n} + \SG_{n,j} \SU_{k,n+j},
\end{align}
where $\SF_{n,0} \triangleq \SF_{n}$, $\SG_{n,0} \triangleq \SG_{n}$, and, for all $j \ge 1$, $\SF_{n,j} \in \BBR^{p \times pn}$ and $\SG_{n,j} \in \BBR^{p \times m(n+1+j)}$ are defined
\begin{align}
    \SF_{n,j} &\triangleq \mathcal{S}(\SF_{n,0},j) - \sum_{i=1}^{\mathclap{\min\{j,n\}}} F_i \SF_{n,j-i} , \label{eqn: SF_n,j defn}
    \\
    \SG_{n,j} &\triangleq \begin{bmatrix} \SG_{n,0} & 0_{p \times jm} \end{bmatrix} - \sum_{i=1}^{\mathclap{\min\{j,n\}}} F_i \begin{bmatrix} 0_{p \times im} & \SG_{n,j-i} \end{bmatrix}, \label{eqn: SG_n,j defn}
\end{align}
and where $\mathcal{S}(\SF_{n,0},j) \in \BBR^{p \times pn}$ is the $j$ step shift of $\SF_{n,0}$, defined as 
\begin{equation}
    \mathcal{S}(\SF_{n,0},j) \triangleq \begin{cases}
        \begin{bmatrix}
            F_{j+1} & \cdots & F_n & 0_{p \times jp} 
        \end{bmatrix}
        & j \le n-1,
        \\
        0_{p \times np} & j \ge n.
    \end{cases}
\end{equation}
\end{prop}

\begin{proof}
    Proof follows by strong induction on $j \ge 0$. Note that for all $k \ge k_0$ and $j = 0$, \eqref{eqn: IO j step update} simplifies to \eqref{eqn: IO matrix form}. 
    Next, let $k \ge k_0$, let $j \ge 1$ and suppose that, for all $\hat{j} \le j -1$, \eqref{eqn: IO j step update} holds.
    Note that by \eqref{eqn: IO Model}, $y_{k+n+j}$ can be written as
    \begin{equation*}
        y_{k+n+j} = -\sum_{i=1}^n F_i y_{k+n+j-i} + \sum_{i=0}^n G_i u_{k+n+j-i}.
    \end{equation*}
    Next, for all $1\le i \le \min\{j,n\}$, it follows that $0 \le j-i \le j-1$ and by inductive hypothesis, $y_{k+n+j-i}$ can be written as
    \begin{align*}
        y_{k+n+j-i} &= -\SF_{n,j-i} \SY_{k,n} + \SG_{n,j-i} \SU_{k,n+j-i} \\
        &= -\SF_{n,j-i} \SY_{k,n} + \begin{bmatrix} 0_{p \times im} & \SG_{n,j-i} \end{bmatrix}  \SU_{k,n+j}.
    \end{align*}
    Furthermore, note that
    \begin{align*}
        \sum_{i=0}^n G_i u_{k+n+j-i} = \begin{bmatrix}
        \SG_{n,0} & 0_{p \times jm}
        \end{bmatrix} \SU_{k,n+j}.
    \end{align*}
    Hence, $y_{k+n+j}$ can be written as
    \begin{align*}
        y_{k+n-j} =&  \sum_{i=1}^{\mathclap{\min\{j,n\}}} F_i \left(\SF_{n,j-i} \SY_{k,n} - \begin{bmatrix} 0_{p \times im} & \SG_{n,j-i} \end{bmatrix}  \SU_{k,n+j}\right)
        \\ 
        & - \hspace{5pt}  \sum_{\mathclap{i = \min\{j,n\}+1}}^n \hspace{5pt} F_i y_{k+n+j-i} 
        + \begin{bmatrix}
        \SG_{n,0} & 0_{p \times jm}
        \end{bmatrix} \SU_{k,n+j} .
    \end{align*}
    Noting that $\sum_{i = \min\{j,n\}+1}^n F_i y_{k+n+j-i} = \mathcal{S}(\SF_{n,0},j) \SY_{k,n}$
    and combining terms yields \eqref{eqn: IO j step update}.
\end{proof}

\section{Equivalence and Reducibility of Input/Output Models}
\label{sec: equivalence and reduciblity}
Consider an input/output model of order $\hat{n} \ge n$ where, for all $k \ge k_0$, $u_k \in \BBR^m$ is the input, $\hat{y}_{k_0},\hdots,\hat{y}_{k_0+\hat{n}-1} \in \BBR^{p}$ are the initial conditions and, for all $k \ge k_0$, the output $\hat{y}_{k+\hat{n}} \in \BBR^p$ is given by
\begin{align}
\label{eqn: IO Model 2}
    \hat{y}_{k+\hat{n}} = -\sum_{i=1}^{\hat{n}} \hat{F}_i \hat{y}_{k+\hat{n}-i} + \sum_{i=0}^{\hat{n}} \hat{G}_i u_{k+\hat{n}-i},
\end{align}
where $\hat{F}_1,\hdots,\hat{F}_{\hat{n}} \in \BBR^{p \times p}$, and $\hat{G}_0,\hdots,\hat{G}_{\hat{n}} \in \BBR^{p \times m}$ are the input/output model coefficients.
Definition \ref{defn: equivalence of models} defines models \eqref{eqn: IO Model} and \eqref{eqn: IO Model 2} as \textit{equivalent} if they give the same outputs under the same inputs and initial conditions.
\begin{defn}
\label{defn: equivalence of models}
    Let $k_0 \in \BBN$.
    Consider input/output model \eqref{eqn: IO Model} with order $n$ and and input/output Model \eqref{eqn: IO Model 2} with order $\hat{n} \ge n$. 
    For all $k \ge k_0$, let $u_k \in \BBR^m$ be arbitrary. For all $k_0 \le k \le k_0 + n-1$, let $y_k \in \BBR^p$ be arbitrary and, for all $k \ge k_0+n$, let $y_k \in \BBR^p$ be given by \eqref{eqn: IO Model}. 
    Next, for all $k_0 \le k \le k_0+\hat{n}-1$, let $\hat{y}_k = y_k$ and, for all $k \ge k_0+\hat{n}$, let $\hat{y}_k \in \BBR^p$ be given by \eqref{eqn: IO Model 2}.
    Models \eqref{eqn: IO Model} and \eqref{eqn: IO Model 2} are \textbf{equivalent} if, for all $k \ge k_0+\hat{n}$, $\hat{y}_k = y_k$. 
\end{defn}

Proposition \ref{prop: equivalence of same order models} shows that two input/output models of the same order are equivalent if and only if they have the same model coefficients.
\begin{prop}
\label{prop: equivalence of same order models}
     Consider input/output models \eqref{eqn: IO Model} and \eqref{eqn: IO Model 2} with the same order $\hat{n} = n$. Models \eqref{eqn: IO Model} and \eqref{eqn: IO Model 2} are equivalent if and only if $\hat{G}_0 = G_0$ and, for all $1 \le i \le n$, $\hat{F}_i = F_i$ and $\hat{G}_i = G_i$.
\end{prop}
\begin{proof}
    Proof of necessity follow immediately. 
    To prove sufficiency, let $k_0 \in \BBN$. First consider that there exists $0 \le i \le n$ such that $\hat{G}_i \ne G_i$. Then, there exists $u \in \BBR^m$ such that $\hat{G}_i u \ne G_i u$.
    Next, consider $y_{k_0} = \cdots = y_{k_0+n-1} = 0$, for all $k \ge k_0$ such that $k \ne k_0+n-i$, $u_k = 0$, and $u_{k_0+n-i} = u$. Then, $\hat{y}_{k_0+n} = \hat{G}_i u$, $y_{k_0+n} = G_i u$, and $\hat{y}_{k_0 +n} \ne y_{k_0+n}$.

    Next, consider that there exists $1 \le i \le n$ such that $\hat{F}_i \ne F_i$. Similarly, there exists $y \in \BBR^p$ such that $\hat{F}_i y \ne F_i y$. Hence, if, for all $k_0 \le k \le k_0+n-1$, $k \ne k_0+n-i$, $y_k = 0$, $y_{k_0+n-i} = y$, and, for all $k \ge k_0$, $u_k = 0$, then $\hat{y}_{k_0+n} = \hat{F}_i y$, $y_{k_0+n} = F_i y$, and $\hat{y}_{k_0+n} \ne y_{k_0+n}$. 
    Thus, \eqref{eqn: IO Model} and \eqref{eqn: IO Model 2} are not equivalent and sufficiency is proven.
\end{proof}

We next consider the equivalence of input/output models of different orders. Note that, for all $k \ge k_0$,
\begin{align}
    \hat{y}_{k+\hat{n}} = -\hat{\SF}_{\hat{n}} \hat{\SY}_{k,\hat{n}} + \hat{\SG}_{\hat{n}} \SU_{k,\hat{n}},
\end{align}
where $\SU_{k,\hat{n}} \in \BBR^{m(\hat{n}+1)}$ is defined in \eqref{eqn: SG and SU definition} and $\hat{\SF}_{\hat{n}} \in \BBR^{p \times p\hat{n}}$, $\hat{\SY}_{k,\hat{n}} \in \BBR^{p\hat{n}}$, and $\hat{\SG}_{\hat{n}} \in \BBR^{p \times m(\hat{n}+1)}$ are defined as
\begin{align}
    \hat{\SF}_{\hat{n}} &\triangleq  \begin{bmatrix}
        \hat{F}_1 & \cdots & \hat{F}_{\hat{n}}
    \end{bmatrix} , \quad 
    \hat{\SY}_{k,\hat{n}} \triangleq \begin{bmatrix}
        \hat{y}_{k+\hat{n}-1} \\ \vdots \\ \hat{y}_k
    \end{bmatrix},
    \\
    \hat{\SG}_{\hat{n}} &\triangleq \begin{bmatrix}
        \hat{G}_0 & \cdots & \hat{G}_{\hat{n}}
    \end{bmatrix}.
\end{align}
Theorem \ref{theo: IO Equivalence condition} gives necessary and sufficient conditions for the equivalence of input/output models of different orders.
We begin with two useful Lemmas. 

\begin{lem}
\label{lem: IO 2 consolidation}
Let $k_0 \in \BBN$. Consider input/output model \eqref{eqn: IO Model} with order $n$ and and input/output Model \eqref{eqn: IO Model 2} with order $\hat{n} > n$. For all $k \ge k_0$, let $u_k \in \BBR^m$, for all $k_0 \le k \le k_0+n-1$, let $y_k \in \BBR^p$, and, for all $k \ge k_0+n$, let $y_k \in \BBR^p$ be given by \eqref{eqn: IO Model}. 
Let $k^* \ge k_0$ and, for all $k_0 \le k \le k^*+\hat{n}-1$, let $\hat{y}_k = y_k$.
Finally, let $\hat{y}_{k^*+\hat{n}} \in \BBR^p$ be given by \eqref{eqn: IO Model 2}. 
Then, $\hat{y}_{k^*+\hat{n}}$ can be expressed as
\begin{align}
    \label{eqn: IO 2 consolidation}
    \hat{y}_{k^*+\hat{n}} = -\hat{\SF}_{n,\hat{n}-n} \SY_{k^*,n} + \hat{\SG}_{n,\hat{n}-n}\SU_{k^*,\hat{n}},
\end{align}
where $\hat{\SF}_{n,\hat{n}-n} \in \BBR^{p \times pn}$ and $\hat{\SG}_{n,\hat{n}-n} \in \BBR^{p \times m (\hat{n}+1)}$ are defined
\begin{align}
    \hat{\SF}_{n,\hat{n}-n} &\triangleq 
    \begin{bmatrix}
        \hat{F}_{\hat{n}-n+1} & \cdots & \hat{F}_{\hat{n}}
    \end{bmatrix} - 
    \sum_{i=1}^{\hat{n}-n} \hat{F}_i \SF_{n,\hat{n} - n - i},
    \label{eqn: SF_n,nhat-n}
    \\
    \hat{\SG}_{n,\hat{n}-n} &\triangleq \hat{\SG}_{\hat{n}} - \sum_{i=1}^{\hat{n}-n} \hat{F}_i 
    \begin{bmatrix}
        0_{p \times im} & \SG_{n,\hat{n} - n - i}
    \end{bmatrix}.
    \label{eqn: SG_n,nhat-n}
\end{align}
\end{lem}

\begin{proof}
    Let $k_0 \in \BBN$ and let $k^* \ge k_0$. By \eqref{eqn: IO Model 2}, $\hat{y}_{k^*+\hat{n}}$ can be written as
    \begin{equation}
    \label{eqn: IO 2 consolidation temp 1}
        \hat{y}_{k^*+\hat{n}} = -\sum_{i=1}^{\hat{n}-n} \hat{F}_i y_{k^*+\hat{n}-i} -\sum_{\mathclap{i=\hat{n}-n + 1}}^{\hat{n}} \hat{F}_i y_{k+\hat{n}-i} + \hat{\SG}_{\hat{n}} \SU_{k^*,\hat{n}}.
    \end{equation}
    Note that the second term can be written as
    \begin{align}
    \label{eqn: IO 2 consolidation temp 2}
        \sum_{i=\hat{n}-n + 1}^{\hat{n}} \hat{F}_i y_{k^*+\hat{n}-i} = 
        \begin{bmatrix}
        \hat{F}_{\hat{n}-n+1} & \cdots & \hat{F}_{\hat{n}}
        \end{bmatrix}
        \SY_{k^*,n}.
    \end{align}
    Next, for all $1 \le i \le \hat{n}-n$, note that $\hat{n} - n - i \ge  0$ and Proposition \ref{prop: IO j step update} (with $k = k^*$ and $j = \hat{n}-n-i$) implies that 
    \begin{align}
    \label{eqn: IO 2 consolidation temp 3}
        y_{k^*+\hat{n}-i} = -\SF_{n,\hat{n}-n-i} \SY_{k^*,n} + \SG_{n,\hat{n}-n-i}\SU_{k^*,\hat{n}-i}.
    \end{align}
    Substituting \eqref{eqn: IO 2 consolidation temp 2} and \eqref{eqn: IO 2 consolidation temp 3} into \eqref{eqn: IO 2 consolidation temp 1}, it follows that
    \begin{align}
        \hat{y}_{k^*+\hat{n}} =& -\sum_{i=1}^{\hat{n}-n} \hat{F}_i (-\SF_{n,\hat{n}-n-i} \SY_{k^*,n} + \SG_{n,\hat{n}-n-i}\SU_{k^*,\hat{n}-i}) \nonumber
        \\
        &-\begin{bmatrix}
        \hat{F}_{\hat{n}-n+1} & \cdots & \hat{F}_{\hat{n}}
        \end{bmatrix}
        \SY_{k^*,n}
        + \hat{\SG}_{\hat{n}} \SU_{k^*,\hat{n}}.
        \label{eqn: IO 2 consolidation temp 4}
    \end{align}
    Finally, note that, for all $1 \le i \le \hat{n} - n$,
    \begin{align}
    \label{eqn: IO 2 consolidation temp 5}
        \SG_{n,\hat{n}-n-i}\SU_{k^*,\hat{n}-i} = 
        \begin{bmatrix}
            0_{p \times im} & \SG_{n,\hat{n} - n - i}
        \end{bmatrix} \SU_{k^*,\hat{n}}.
    \end{align}
    Substituting \eqref{eqn: IO 2 consolidation temp 4} into \eqref{eqn: IO 2 consolidation temp 5} and simplifying yields \eqref{eqn: IO 2 consolidation}. 
\end{proof}

\begin{lem}
\label{lem: connecting IO 1 and IO 2}
For all $\hat{n} > n$, 
    \begin{align}
    \label{eqn: connecting IO 1 and IO 2}
        \begin{bmatrix}
            -\hat{\SF}_{n,\hat{n}-n} & \hat{\SG}_{n,\hat{n}-n}
        \end{bmatrix}
        =
        \begin{bmatrix}
            -\hat{\SF}_{\hat{n}} & \hat{\SG}_{\hat{n}}
        \end{bmatrix} M_{n,\hat{n}}, 
    \end{align}
    where $M_{n,\hat{n}} \in \BBR^{p\hat{n} + m(\hat{n}+1) \times pn + m(\hat{n}+1)}$ is defined
    \begingroup 
    \setlength\arraycolsep{2pt}
    \begin{equation}
    \label{eqn: M defn}
        M_{n,\hat{n}} \triangleq \begin{bmatrix}
            M_{n,\hat{n}}^1 & M_{n,\hat{n}}^2 
            \\
            I_{pn} & 0_{pn \times m (\hat{n}+1)}
            \\ 
            0_{m(\hat{n}+1) \times pn} & I_{m(\hat{n}+1)}
        \end{bmatrix}
        =
        \begin{bmatrix}
            \begin{bmatrix} M_{n,\hat{n}}^1 & M_{n,\hat{n}}^2 \end{bmatrix}
            \\
            I_{pn + m(\hat{n}+1)}
        \end{bmatrix},
    \end{equation}
    \endgroup
    where $M_{n,\hat{n}}^1 \in \BBR^{p (\hat{n}-n) \times pn}$ and $M_{n,\hat{n}}^2 \in \BBR^{p (\hat{n}-n) \times m (\hat{n}+1)}$ are defined
    \begin{align}
        M_{n,\hat{n}}^1 \triangleq \begin{bmatrix}
            -\SF_{n,\hat{n}-n-1} \\ -\SF_{n,\hat{n}-n-2} \\ \vdots \\ -\SF_{n,1} \\ -\SF_{n,0} 
        \end{bmatrix}
        , 
        M_{n,\hat{n}}^2 \triangleq \begin{bmatrix}
            \begin{bmatrix} 0_{p \times m} & \SG_{n,\hat{n} - n - 1} \end{bmatrix}
            \\
            \begin{bmatrix} 0_{p \times 2m} & \SG_{n,\hat{n} - n - 2} \end{bmatrix}
            \\
            \vdots 
            \\
            \begin{bmatrix} 0_{p \times (\hat{n}-n-1)m} & \SG_{n,1} \end{bmatrix}
            \\
            \begin{bmatrix} 0_{p \times (\hat{n}-n)m} & \SG_{n,0} \end{bmatrix}
        \end{bmatrix}
    \end{align}
\end{lem}

\begin{proof}
    It follows from \eqref{eqn: SF_n,nhat-n} that
    \begin{align*}
        \hat{\SF}_{n,\hat{n}-n} = \begin{bmatrix}
        \hat{F}_{\hat{n}-n+1} & \cdots & \hat{F}_{\hat{n}}
    \end{bmatrix} - 
    \begin{bmatrix}
        \hat{F}_{1} & \cdots & \hat{F}_{\hat{n}-n}
    \end{bmatrix} M_{n,\hat{n}}^1.
    \end{align*}
    Furthermore, it follows from \eqref{eqn: SG_n,nhat-n} that
    \begin{align*}
        \hat{\SG}_{n,\hat{n}-n} = \hat{\SG}_{\hat{n}} - \begin{bmatrix}
        \hat{F}_{1} & \cdots & \hat{F}_{\hat{n}-n}
    \end{bmatrix} M_{n,\hat{n}}^2.
    \end{align*}
    Combining these two expressions can yields \eqref{eqn: connecting IO 1 and IO 2}.
\end{proof}

\begin{theo}
\label{theo: IO Equivalence condition}
    Consider input/output model \eqref{eqn: IO Model} with order $n$ and input/output model \eqref{eqn: IO Model 2} with order $\hat{n} > n$. Models \eqref{eqn: IO Model} and \eqref{eqn: IO Model 2} are equivalent if and only if 
    \begin{align}
    \label{eqn: IO Equivalence condition v2}
        \hat{\SF}_{n,\hat{n}-n} = \SF_{n,\hat{n}-n}, \quad \hat{\SG}_{n,\hat{n}-n} = \SG_{n,\hat{n}-n},
    \end{align}
    which holds if and only if
    \begin{align}
    \label{eqn: IO Equivalence condition}
        \begin{bmatrix}
            -\hat{\SF}_{\hat{n}} & \hat{\SG}_{\hat{n}}
        \end{bmatrix} M_{n,\hat{n}}
        =
        \begin{bmatrix}
            -{\SF}_{n,\hat{n}-n} & {\SG}_{n,\hat{n}-n}
        \end{bmatrix}.
    \end{align}
\end{theo}

\begin{proof}
    It follows from Lemma \ref{lem: connecting IO 1 and IO 2} that \eqref{eqn: IO Equivalence condition v2} and \eqref{eqn: IO Equivalence condition} are equivalent.
    Hence, it suffices to show that Models \eqref{eqn: IO Model} and \eqref{eqn: IO Model 2} are equivalent if and only if \eqref{eqn: IO Equivalence condition v2} holds. 
    To prove sufficiency, assume models \eqref{eqn: IO Model} and \eqref{eqn: IO Model 2} are equivalent and let $k_0 \in \BBN$. 
    Then, it follows from Definition \ref{defn: equivalence of models} that, for all $\SY_{k_0,n} \in \BBR^{pn}$ and $\SU_{k_0,\hat{n}} \in \BBR^{m(\hat{n}+1)}$, $y_{k_0+\hat{n}} = \hat{y}_{k_0+\hat{n}}$ holds, where, for all $k_0 + n \le k \le k_0 + \hat{n}$,
    $y_{k} \in \BBR^p$ is given by \eqref{eqn: IO Model}, for all $k_0 \le k \le k_0+\hat{n}-1$, $\hat{y}_{k} = y_{k}$, and where $\hat{y}_{k+\hat{n}} \in \BBR^p$ is given by \eqref{eqn: IO Model 2}.
    It follows from Proposition \ref{prop: IO j step update} and Lemma \ref{lem: IO 2 consolidation} that
    \begin{align*}
        y_{k + \hat{n}} &= -{\SF}_{n,\hat{n}-n} \SY_{k,n} + {\SG}_{n,\hat{n}-n} \SU_{k,n},
        \\
        \hat{y}_{k + \hat{n}} &= -\hat{\SF}_{n,\hat{n}-n} \SY_{k,n} + \hat{\SG}_{n,\hat{n}-n} \SU_{k,n},
    \end{align*}
    Since $\SY_{k,n}$ and $\SU_{k,\hat{n}}$ are chosen arbitrarily and $y_{k+\hat{n}} = \hat{y}_{k+\hat{n}}$, \eqref{eqn: IO Equivalence condition v2} follows.

    To prove necessity, assume that \eqref{eqn: IO Equivalence condition} holds, let $k_0 \in \BBN$. For all $k \ge k_0$, let $u_k \in \BBR^m$ be arbitrary, for all $k_0 \le k \le k_0 + n-1$, let $y_k \in \BBR^p$ be arbitrary and, for all $k \ge k_0 + n$, let $y_k \in \BBR^p$ be given by \eqref{eqn: IO Model}. 
    Furthermore, for all $k_0 \le k \le k_0+\hat{n}-1$, let $\hat{y}_k = y_k$ and, for all $k \ge k_0+\hat{n}$, let $\hat{y}_k \in \BBR^p$ be given by \eqref{eqn: IO Model 2}.
    We now show by strong induction that, for all $k^* \ge k_0$, $\hat{y}_{k^*+\hat{n}} = y_{k^*+\hat{n}}$.
    First, consider the base case $k^* = k_0$. 
    It follows from \eqref{eqn: IO j step update} of Proposition \ref{prop: IO j step update} (with $k = k_0$ and $j = \hat{n}-n$) that $y_{k_0+\hat{n}} = -\SF_{n,\hat{n}-n} \SY_{k_0,n} + \SG_{n,\hat{n}-n} \SU_{k_0,\hat{n}}$.
    Moreover, since by assumption, for all $k_0 \le k \le k_0+\hat{n}-1$, $\hat{y}_k = y_k$,
    it follows from \eqref{eqn: IO 2 consolidation} of Lemma \ref{lem: IO 2 consolidation} (with $k^* = k_0$) that $\hat{y}_{k_0+\hat{n}} = -\hat{\SF}_{n,\hat{n}-n} \SY_{k_0,n} + \hat{\SG}_{n,\hat{n}-n}\SU_{k_0,\hat{n}}$.
    Finally, \eqref{eqn: IO Equivalence condition v2} implies that $y_{k_0+\hat{n}} = \hat{y}_{k_0+\hat{n}}$.

    Next, let $k^* \ge k_0+1$ and assume for inductive hypothesis that, for all $k_0 \le k \le k^*-1$, $\hat{y}_{k+\hat{n}} = y_{k+\hat{n}}$.
    It follows from \eqref{eqn: IO j step update} of Proposition \ref{prop: IO j step update} (with $k = k^*$ and $j = \hat{n}-n$) that $y_{k^* + \hat{n}} = -\SF_{n,\hat{n}-n} \SY_{k^*,n} + \SG_{n,\hat{n}-n} \SU_{k^*,\hat{n}}$.
    Moreover, it follows from inductive hypothesis and \eqref{eqn: IO 2 consolidation} of Lemma \ref{lem: IO 2 consolidation} that $\hat{y}_{k^* + \hat{n}} = -\hat{\SF}_{n,\hat{n}-n} \SY_{k^*,n} + \hat{\SG}_{n,\hat{n}-n}\SU_{k^*,\hat{n}}$.
    Finally, \eqref{eqn: IO Equivalence condition v2} implies that $y_{k^*+\hat{n}} = \hat{y}_{k^*+\hat{n}}$.
    Thus, by strong induction, for all $k^* \ge k_0$, $\hat{y}_{k^*+\hat{n}} = y_{k^*+\hat{n}}$, and models \eqref{eqn: IO Model} and \eqref{eqn: IO Model 2} are equivalent.
\end{proof}
Note that the matrix $M_{n,\hat{n}} \in \BBR^{p\hat{n} + m(\hat{n}+1) \times pn + m(\hat{n}+1)}$ will be of importance in the later discussion on online identification.
\begin{example}
To better illustrate the result of Theorem \ref{theo: IO Equivalence condition}, consider the case $m = p = 1$, $n = 1$, and $\hat{n} = 2$.
It follows from Theorem \ref{theo: IO Equivalence condition} that models \eqref{eqn: IO Model} and \eqref{eqn: IO Model 2} are equivalent if and only if 
\begin{align}
    \begin{bmatrix}
        -\hat{F}_1 & -\hat{F}_2 & \hat{G}_0 & \hat{G}_1 & \hat{G}_2
    \end{bmatrix}
    \begin{bmatrix}
        -F_1 & 0 & G_0 & G_1 
        \\
        1 & 0 & 0 & 0 
        \\
        0 & 1 & 0 & 0 
        \\
        0 & 0 & 1 & 0 
        \\
        0 & 0 & 0 & 1 
    \end{bmatrix} \nonumber
    \\
    =
    \begin{bmatrix}
        F_1^2 & G_0 & G_1 - F_1 G_0 & -F_1 G_1
    \end{bmatrix}.
    \label{eqn: 1st and 2nd order equivalence}
\end{align}
As a sanity check, note that the trivial case $\hat{F}_1 = F_1$, $\hat{G}_0 = G_0$, $\hat{G}_1 = G_1$ and $\hat{F}_2 = \hat{G}_2 = 0$ satisfies \eqref{eqn: 1st and 2nd order equivalence}. $\hfill\mbox{$\diamond$}$
\end{example}

Definition \ref{defn: reducible} defines an input/output model as \textit{reducible} if there exists an equivalent input/output model of lower-order, and \textit{irreducible} otherwise.
Theorem \ref{theo: reducibility} provides necessary and sufficient conditions for reducibility of input/output models.

\begin{defn}
\label{defn: reducible}
    Consider input/output model \eqref{eqn: IO Model 2} with order $\hat{n}$. The input/output model \eqref{eqn: IO Model 2} is \textbf{reducible} if there exists an input/output model \eqref{eqn: IO Model} with order $n \le \hat{n}-1$ such that $\eqref{eqn: IO Model}$ and \eqref{eqn: IO Model 2} are equivalent.
    Otherwise, input/output model \eqref{eqn: IO Model 2} is \textbf{irreducible}.
\end{defn}

\begin{theo}
\label{theo: reducibility}
The input/output model \eqref{eqn: IO Model 2} is reducible if and only if there exists $F_1 \in \BBR^{p \times p}$ such that
\begin{align}
(\hat{F}_1 - F_1) F_{\hat{n}-1} &= \hat{F}_{\hat{n}},
\label{eqn: reducible condition 1}
\\
(\hat{F}_1 - F_1)  G_{\hat{n}-1} &= \hat{G}_{\hat{n}},
\label{eqn: reducible condition 2}
\end{align}
where $G_0 \triangleq \hat{G}_0$ and, for all $2 \le i \le \hat{n} - 1$ and $1 \le j \le \hat{n} - 1$,
\begin{align}
    F_{i} \triangleq &= \hat{F}_{i} - (\hat{F}_1 - F_1) F_{i-1},
    \label{eqn: reduced coefficients 1}
    \\
    G_{j} \triangleq  &= \hat{G}_{j} - (\hat{F}_1 - F_1) G_{j-1}.
    \label{eqn: reduced coefficients 2}
\end{align}
Moreover, if there exists $F_1 \in \BBR^{p \times p}$ satisfying \eqref{eqn: reducible condition 1} and \eqref{eqn: reducible condition 2}, then input/output model \eqref{eqn: IO Model} with coefficients $F_1,\hdots,F_{\hat{n}-1}$ and $G_0,\hdots,G_{\hat{n}-1}$ defined by \eqref{eqn: reduced coefficients 1} and \eqref{eqn: reduced coefficients 2} is equivalent to input/output model \eqref{eqn: IO Model 2}.
\end{theo}
\begin{proof}
    To begin, if input/output model \eqref{eqn: IO Model} with order $n < \hat{n}-1$ is equivalent to input/output model \eqref{eqn: IO Model 2}, then there exists an input/output model with order $\hat{n}-1$ that is equivalent to \eqref{eqn: IO Model 2}, namely a model with the coefficients $F_1,\hdots,F_n$, $G_0,\hdots,G_n$, and $F_i = 0_{p \times p}$ and $G_i = 0_{p \times m}$ for all $n+1 \le i \le \hat{n}-1$.
    Hence, input/output model \eqref{eqn: IO Model 2}, with order $\hat{n}$ is reducible if and only if there exists an equivalent input/output model \eqref{eqn: IO Model} with order $n = \hat{n}-1$.

    Next, it follows from Theorem \ref{theo: IO Equivalence condition} that model \eqref{eqn: IO Model} with order $n = \hat{n}-1$ and model \eqref{eqn: IO Model 2} with order $\hat{n}$ are equivalent if and only if $\hat{\SF}_{\hat{n}-1,1} = \SF_{\hat{n}-1,1}$ and $\hat{\SG}_{\hat{n}-1,1} = \SG_{\hat{n}-1,1}$.
    Moreover, it follows from \eqref{eqn: SF_n,j defn}, \eqref{eqn: SF_n,nhat-n}, \eqref{eqn: SG_n,j defn}, and \eqref{eqn: SG_n,nhat-n} that
    \begin{align*}
        \hat{\SF}_{\hat{n}-1,1} & = \begin{bmatrix}
            \hat{F}_2  & \hat{F}_3 & \cdots & \hat{F}_{\hat{n}}
        \end{bmatrix} - \hat{F}_1 \SF_{\hat{n}-1} ,
        \\
        \SF_{\hat{n}-1,1} & =
        -F_1 \SF_{\hat{n}-1} + \begin{bmatrix}
            \hat{F}_2  & \hat{F}_3 & \cdots & \hat{F}_{\hat{n}-1} & 0_{p \times p}
        \end{bmatrix},
        \\
       \hat{\SG}_{\hat{n}-1,1} &=  \hat{\SG}_{\hat{n}} - \hat{F}_1 \begin{bmatrix}
            0_{p \times m} & \SG_{\hat{n}-1}
        \end{bmatrix} ,
        \\
        \SG_{\hat{n}-1,1} & =
        -F_1 \begin{bmatrix}
            0_{p \times m} & \SG_{\hat{n}-1}
        \end{bmatrix}
        +
        \begin{bmatrix}
            \SG_{\hat{n}-1} & 0_{p \times m}
        \end{bmatrix}.
    \end{align*}
    Rearranging terms, it follows that $\hat{\SF}_{\hat{n}-1,1} = \SF_{\hat{n}-1,1}$ and $\hat{\SG}_{\hat{n}-1,1} = \SG_{\hat{n}-1,1}$ hold if and only if 
    \begin{align*}
        & (\hat{F}_1 - F_1) \SF_{\hat{n}-1} = \begin{bmatrix}
            \hat{F}_2 - F_2 & \cdots & \hat{F}_{\hat{n}-1} - F_{\hat{n}-1} & \hat{F}_{\hat{n}}
        \end{bmatrix},
        \\
        & (\hat{F}_1 - F_1) \begin{bmatrix}
            0_{p \times m} & \SG_{\hat{n}-1}
        \end{bmatrix}
        \\
        & \hphantom{(\hat{F}_1 - F_1) \SF_{\hat{n}-1}} =
        \begin{bmatrix}
            \hat{G}_0 - G_0 & \cdots & \hat{G}_{\hat{n}-1} - G_{\hat{n}-1} & \hat{G}_{\hat{n}}
        \end{bmatrix},
    \end{align*}
    which can be expanded as, for all $2 \le i \le \hat{n} - 1$ and $1 \le j \le \hat{n} - 1$,
    \begin{align}
        (\hat{F}_1 - F_1) F_{i-1} &= \hat{F}_{i} - F_{i}, \quad , \label{eqn: reducible temp 1}
        \\
        (\hat{F}_1 - F_1) F_{\hat{n}-1} &= \hat{F}_{\hat{n}}, \label{eqn: reducible temp 2}
        \\
        0_{p \times m} & = \hat{G}_0 - G_0. \label{eqn: reducible temp 3}
        \\
        (\hat{F}_1 - F_1) G_{j-1} &= \hat{G}_{j} - G_{j}, \label{eqn: reducible temp 4}
        \\
        (\hat{F}_1 - F_1)  G_{\hat{n}-1} &= \hat{G}_{\hat{n}}. \label{eqn: reducible temp 5}
    \end{align}
    Finally, note that, for all $2 \le i \le \hat{n}-1$, \eqref{eqn: reducible temp 1} holds if and only if \eqref{eqn: reduced coefficients 1}, for all $1 \le j \le \hat{n}-1$, \eqref{eqn: reducible temp 4} holds if and only if \eqref{eqn: reduced coefficients 2}, and \eqref{eqn: reducible temp 3} holds if and only if $ G_0= \hat{G}_0$.
    Hence, model \eqref{eqn: IO Model} with order $n = \hat{n}-1$ and model \eqref{eqn: IO Model 2} with order $\hat{n}$ are equivalent if and only if $F_1$ satisfies \eqref{eqn: reducible temp 2} and \eqref{eqn: reducible temp 5}.
\end{proof}

\section{Online Identification of Input/Output Models using Recursive Least Squares}
\label{sec: online ID of IO Models}

Next, we discuss the online identification of input/output models using recursive least squares. 
We consider input/output model \eqref{eqn: IO Model} of order $n$ to be the true system. 
We define $\theta_{n,{\rm true}} \in \BBR^{p \times pn + m(n+1)}$ as
\begin{equation}
    \theta_{n,{\rm true}} \triangleq \begin{bmatrix}
        F_1 & \cdots & F_{n} & G_0 & \cdots & G_{n} 
    \end{bmatrix} = \begin{bmatrix}
        \SF_n & \SG_n
    \end{bmatrix}.
\end{equation}
It follows from \eqref{eqn: IO matrix form} that, for all $k \ge 0$, $y_k \in \BBR^p$ is given by
\begin{align}
\label{eqn: yk = thetatrue phi}
    y_k = \theta_{n,{\rm true}} \phi_{n,k},
\end{align}
where, for all $k \ge 0$ and $n \ge 0$, $\phi_{n,k} \in \BBR^{pn + m(n+1)}$ is defined as
\begin{align}
\label{eqn: phik defn}
    \phi_{n,k} \triangleq \begin{bmatrix}
        -y_{k-1} \\ \vdots \\ -y_{k-n} \\ u_k \\ \vdots \\ u_{k-n}
    \end{bmatrix}
    = \begin{bmatrix}
        \SY_{k-n,n} \\ \SU_{k-n,n}
    \end{bmatrix}.
\end{align}

The objective of online identification is to identify the coefficients an input/output model \eqref{eqn: IO Model 2} of order $\hat{n}$ using measurements of the inputs $u_k$ and outputs $y_k$ generated from \eqref{eqn: IO Model}. 
This can be accomplished by minimizing the cost function $J_{\hat{n},k} \colon \BBR^{p \times p \hat{n}+ m(\hat{n}+1)} \rightarrow \BBR$, defined as
\begin{align}
\label{eqn: PCAC cost}
    J_{\hat{n},k}(\hat{\theta}_{\hat{n}}) = & \sum_{i=0}^k z_{\hat{n},i}^\rmT(\hat{\theta}_{\hat{n}}) z_{\hat{n},i}(\hat{\theta}_{\hat{n}}) \nonumber
    \\
    & + \tr \left[(\hat{\theta}_{\hat{n}}-{\theta}_{\hat{n},0}) {P}_{\hat{n},0}^{-1} (\hat{\theta}_{\hat{n}}-{\theta}_{\hat{n},0})^\rmT \right],
\end{align}
where $\hat{\theta}_{\hat{n}} \in \BBR^{p \times p \hat{n}+ m(\hat{n}+1)}$ are the coefficients to be identified, defined as
\begin{align}
    \hat{\theta}_{\hat{n}} \triangleq \begin{bmatrix}
        \hat{F}_1 & \cdots & \hat{F}_{\hat{n}} & \hat{G}_0 & \cdots & \hat{G}_{\hat{n}} 
    \end{bmatrix} = \begin{bmatrix}
        \hat{\SF}_{\hat{n}} & \hat{\SG}_{\hat{n}}
    \end{bmatrix},
\end{align}
where the residual error function $z_{\hat{n},k} \colon \BBR^{p \times p \hat{n}+ m(\hat{n}+1)} \rightarrow \BBR^p$ is defined as
%
%
\begin{align}
\label{eqn: zk = y - theta phi}
    z_{\hat{n},k}(\hat{\theta}_{\hat{n}}) \triangleq y_k - \hat{\theta}_{\hat{n}} \phi_{\hat{n},k},
\end{align}
where $\phi_{\hat{n},k} \in \BBR^{p \hat{n}+ m(\hat{n}+1)}$ is defined in \eqref{eqn: phik defn}, and where $\theta_{\hat{n},0} \in \BBR^{p \times p \hat{n}+ m(\hat{n}+1)}$ is an initial guess of the coefficients and $\bar{P}_{\hat{n},0}^{-1} \in \BBR^{[p^2 \hat{n}+ pm(\hat{n}+1)] \times [p^2 \hat{n}+ pm(\hat{n}+1)]}$ is the positive-definite regularization matrix. 
The following algorithm from \cite{nguyen2021predictive} uses recursive least squares to minimize $J_{\hat{n},k}$.

\begin{prop}
\label{prop: PCAC RLS}
For all $k \ge -\hat{n}$, let $u_k \in \BBR^m$, $y_k \in \BBR^p$.
Furthermore, let $\theta_{\hat{n},0} \in \BBR^{p \times p \hat{n}+ m(\hat{n}+1)}$ and let ${P}_{\hat{n},0} \in \BBR^{[p \hat{n}+ m(\hat{n}+1)] \times [p \hat{n}+ m(\hat{n}+1)]}$ be positive definite. 
Then, for all $k \ge 0$, $J_{\hat{n},k}$, defined in \eqref{eqn: PCAC cost}, has a unique global minimizer, denoted
\begin{align}
    \theta_{\hat{n},k+1} \triangleq \argmin_{\hat{\theta}_{\hat{n}} \in \BBR^{p \times \hat{n}(m+p) + m}} J_k(\hat{\theta}_{\hat{n}}).
\end{align}
which, for all $k \ge 1$, is given by
\begin{align}
\label{eqn: batch IO ID}
    \theta_{\hat{n},k} = (\SY_{0,k} {\Phi}_{\hat{n},k}^\rmT + \theta_{\hat{n},0} {P}_{\hat{n},0}^{-1} ) ({\Phi}_{\hat{n},k} {\Phi}_{\hat{n},k}^\rmT + {P}_{\hat{n},0}^{-1})^{-1},
\end{align}
and where ${\Phi}_{\hat{n},k} \in \BBR^{p \hat{n}+ m(\hat{n}+1) \times k}$ is defined
\begin{align}
    {\Phi}_{\hat{n},k} \triangleq \begin{bmatrix}
        \phi_{\hat{n},k-1} & \cdots & \phi_{\hat{n},0}
    \end{bmatrix}.
\end{align}
Moreover, for all $k \ge 0$, $\bar{\theta}_{\hat{n},k+1}$ is given recursively as
\begin{align}
    {P}_{\hat{n},k+1} &= {P}_{\hat{n},k} - \frac{{P}_{\hat{n},k}{\phi}_{\hat{n},k} {\phi}_{\hat{n},k}^\rmT{{P}_{\hat{n},k}}}
    {1 + {\phi}_{\hat{n},k}^\rmT {P}_{\hat{n},k} {\phi}_{\hat{n},k}} , \label{eqn: P update IO}
    \\
    {\theta}_{\hat{n},k+1} &= {\theta}_{\hat{n},k} + (y_k - {\theta}_{\hat{n},k} {\phi}_{\hat{n},k}) {\phi}_{\hat{n},k}^\rmT  {P}_{\hat{n},k+1}. \label{eqn: theta update IO}
\end{align}
\end{prop}
\begin{proof}
    See \cite{nguyen2021predictive} and \cite{islam2019recursive}.
\end{proof}
In practice, the recursive formulation \eqref{eqn: theta update IO} is used to update the coefficient identification in real time as new measurements are obtained.
We've included the batch formulation \eqref{eqn: batch IO ID} to aid in subsequent analysis.
It will also be beneficial for analysis to note that using Lemma \ref{lem: matrix inversion lemma}, for all $k \ge 0$, \eqref{eqn: P update IO} can be rewritten as
\begin{align}
    {P}_{\hat{n},k+1}^{-1} &= {P}_{\hat{n},k}^{-1} + {\phi}_{\hat{n},k} {\phi}_{\hat{n},k}^\rmT.
    \label{eqn: IO Pinv update}
\end{align}

\subsection{Convergence with Correct Model Order}

We begin by considering the case $\hat{n} = n$, where the correct model order is known.
In this case, the identified coefficients ${\theta}_{n,k}$ and true model coefficients ${\theta}_{n,{\rm true}}$ are the same dimension, and it is natural to define the estimation error $\tilde{\theta}_{n,k} \in \BBR^{p \times p n+ m(n+1)}$ as
\begin{align}
    \tilde{\theta}_{n,k} \triangleq {\theta}_{n,k} - {\theta}_{n,{\rm true}}.
\end{align}
It then follows from \eqref{eqn: yk = thetatrue phi}, \eqref{eqn: theta update IO}, and \eqref{eqn: IO Pinv update} that, for all $k \ge 0$, the estimation error dynamics can be written as
\begin{align}
\label{eqn: correct order dynamics}
    \tilde{\theta}_{{n},k+1} = \tilde{\theta}_{{n},k} {P}_{\hat{n},k}^{-1} {P}_{\hat{n},k+1} .
\end{align}
It then follows from \eqref{eqn: IO Pinv update} that, for all $k \ge 0$,
\begin{align}
\label{eqn: theta tilde cumulative}
    \tilde{\theta}_{{n},k} = \tilde{\theta}_{{n},0} {P}_{{n},0}^{-1} {P}_{{n},k} 
    = \tilde{\theta}_{{n},0} {P}_{{n},0}^{-1} ({\Phi}_{{n},k-1} {\Phi}_{{n},k-1}^\rmT + {P}_{{n},0}^{-1})^{-1}.
\end{align}
We now define the notions of weak persistent excitation and persistent excitation. Note that persistent excitation implies weak persistent excitation.
\begin{defn}
    $(\phi_k)_{k=0}^\infty \subset \BBR^{p \times n}$ is \textbf{weakly persistently exciting} if
    \begin{align}
        \lim_{k \rightarrow \infty} \eigmin \left[ \sum_{i=0}^k \phi_i^\rmT \phi_i \right] = \infty.
    \end{align}
    $(\phi_k)_{k=0}^\infty \subset \BBR^{p \times n}$ is \textbf{persistently exciting} if
    \begin{align}
        C \triangleq \lim_{k \rightarrow \infty} \frac{1}{k} \sum_{i=0}^{k-1} \phi_i^\rmT \phi_i
    \end{align}
    exists and is positive definite. 
\end{defn}

Theorem \ref{theo: IO GAS} shows that weak persistent excitation is necessary and sufficient conditions for the global asymptotic stability (GAS) of the error dynamics \eqref{eqn: correct order dynamics}.
For definition of GAS and further discussion on weak persistent excitation, see \cite{bruce2021necessary}.
Moreover, Theorem \ref{theo: IO GAS} shows that persistent excitation implies that the convergence of $\tilde{\theta}_{n,k}$ to zero is asymptotically proportional to $\nicefrac{1}{k}$.

\begin{theo}
\label{theo: IO GAS}
    Consider the assumptions and notation of Proposition \ref{prop: PCAC RLS}.
    \eqref{eqn: correct order dynamics} is GAS if and only of $(\phi_{n,k}^\rmT)_{k=0}^\infty$ is weakly persistently exciting.
    Moreover, if $(\phi_{n,k}^\rmT)_{k=0}^\infty$ is persistently exciting, then
    \begin{align}
        \lim_{k \rightarrow \infty} k \tilde{\theta}_{n,k} = \tilde{\theta}_{{n},0} {P}_{{n},0}^{-1} C_n^{-1},
    \end{align}
    where $C_n \triangleq \lim_{N \rightarrow \infty} \frac{1}{N} \sum_{i=0}^{N-1} {\phi}_{{n},k-1} {\phi}_{{n},k-1}^\rmT$.
\end{theo}

\noindent {\it Proof.}
    Note that, for all $1 \le i \le p$, 
    \begin{align}
    \label{eqn: ith row error dynamics}
        (\tilde{\theta}_{\hat{n},k+1}^{i})^\rmT = {P}_{\hat{n},k+1} {P}_{\hat{n},k}^{-1}  (\tilde{\theta}_{\hat{n},k}^{i})^\rmT,
    \end{align}
    where $\tilde{\theta}_{\hat{n},k}^{i} \in \BBR^{1 \times pn + m(n+1)}$ is the $i^{\rm th}$ row of $\tilde{\theta}_{\hat{n},k}$.
    It follows from Theorem 3 of \cite{bruce2021necessary} that \eqref{eqn: ith row error dynamics} is GAS if and only of $(\phi_{n,k}^\rmT)_{k=0}^\infty$ is weakly persistently exciting. 
    Hence, \eqref{eqn: correct order dynamics} is GAS if and only of $(\phi_{n,k}^\rmT)_{k=0}^\infty$ is weakly persistently exciting.
    Next, if $(\phi_{n,k}^\rmT)_{k=0}^\infty$ is persistently exciting, it follows from \eqref{eqn: theta tilde cumulative} that
    \begin{align*}
        \lim_{k \rightarrow \infty} k \tilde{\theta}_{{n},k}
    & = \lim_{k \rightarrow \infty} \tilde{\theta}_{{n},0} {P}_{{n},0}^{-1} (\frac{1}{k}{\Phi}_{{n},k-1} {\Phi}_{{n},k-1}^\rmT + \frac{1}{k}{P}_{{n},0}^{-1})^{-1}
    \\
    & = \tilde{\theta}_{{n},0} {P}_{{n},0}^{-1} C_n^{-1}. \tag*{\mbox{$\square$}} 
    \end{align*}
%
Proposition \ref{prop: reducible is necessary} show that for $(\phi_{n,k}^\rmT)_{k=0}^\infty$ to be weakly persistently exciting, it is necessary that \eqref{eqn: IO Model} be irreducible.
\begin{prop}
\label{prop: reducible is necessary}
    If \eqref{eqn: IO Model} is reducible, then $(\phi_{n,k}^\rmT)_{k=0}^\infty$ is not weakly persistently exciting.
\end{prop}
\begin{proof}
    If \eqref{eqn: IO Model} is reducible, then there exists an equivalent input/output model of order $n_{\rm r} < n$. Lemma \ref{lem: phi = M*phi} implies that
    \begin{align*}
        \phi_{n,k} = M_{n_{\rm r},n} \begin{bmatrix}
        \SY_{k-n,n_{\rm r}} \\ \SU_{k-n,n}
    \end{bmatrix}.
    \end{align*}
    Hence, for all $N \ge 0$,
    \begin{align*}
        \rank(\sum_{k=0}^N \phi_{n,k} \phi_{n,k}^\rmT ) \le pn_{\rm r} + m(n+1).
    \end{align*}
    but $\sum_{k=0}^N \phi_{n,k} \phi_{n,k}^\rmT \in \BBR^{pn + m(n+1) \times pn + m(n+1)}$. Therefore, for all $N \ge 0$, $\eigmin \left[ \sum_{k=0}^N \phi_{n,k} \phi_{n,k}^\rmT \right] = 0$.
\end{proof}
Note that if \eqref{eqn: IO Model} is reducible, then there exists an equivalent input/output model of a lower-order which is irreducible.
In other words, \eqref{eqn: IO Model} being reducible can be viewed as the order of the identified model being higher than the order of the true model.
This case is addressed in the following subsection.

\subsection{Convergence with Higher Model Order}

Next, we address the case $\hat{n} > n$, where the identified model order is higher than the model order of the true system.
To begin, note that a model of order $\hat{n}$ which is equivalent to \eqref{eqn: IO Model} is given by the coefficients\footnote{The trivial equivalent model \eqref{eqn: theta nhat true} is chosen because it can be written in terms of only the true model coefficients.}
\begin{align}
\label{eqn: theta nhat true}
    \theta_{\hat{n},{\rm true}} = \begin{bmatrix}
        \SF_n & 0_{p \times p(\hat{n}-n)} & \SG_n & 0_{p \times m(\hat{n}-n)}
    \end{bmatrix}.
\end{align}
In particular, it holds that, for all $k \ge 0$, 
\begin{align}
\label{eqn: yk = thetatrue nhat phi}
    y_k = \theta_{\hat{n},{\rm true}} \phi_{\hat{n},k}.
\end{align}
Since $\hat{n} > n$, it follows from Proposition \ref{prop: reducible is necessary} that $(\phi_{\hat{n},k}^\rmT)_{k=0}^\infty$ is not weakly persistently exciting. 
However, it is possible that $(\phi_{n,\hat{n},k}^\rmT)_{k=0}^\infty$ is weakly persistently exciting, where, for all $k \ge 0$, $\phi_{n,\hat{n},k} \in \BBR^{pn + m(\hat{n}+1)}$ is defined
\begin{align}
\label{eqn: phik n nhat defn}
    \phi_{n,\hat{n},k} \triangleq \begin{bmatrix}
        -y_{k-\hat{n}+n-1} \\ \vdots \\ -y_{k-\hat{n}} \\ u_k \\ \vdots \\ u_{k-\hat{n}}
    \end{bmatrix}
    = \begin{bmatrix}
        \SY_{k-\hat{n},n} \\ \SU_{k-\hat{n},\hat{n}}
    \end{bmatrix}.
\end{align}
By similar reasoning to Proposition \ref{prop: reducible is necessary}, for $(\phi_{n,\hat{n},k}^\rmT)_{k=0}^\infty$ to be weakly persistently exciting, it is necessary that \eqref{eqn: IO Model} be irreducible.
Lemma \ref{lem: phi = M*phi} shows that $\phi_{\hat{n},k} = M_{n,\hat{n}} \phi_{n,\hat{n},k}$ where $M_{n,\hat{n}}$ is defined in \eqref{eqn: M defn}.
\begin{lem}
For all $k \ge 0$ and $\hat{n} > n$, 
\label{lem: phi = M*phi}
    \begin{align}
    \label{eqn: phi = M*phi}
        \phi_{\hat{n},k} &= M_{n,\hat{n}} \phi_{n,\hat{n},k}.
    \end{align}
\end{lem}
\begin{proof}
Note that $\phi_{\hat{n},k}$ can be written as
    \begin{align*}
        \phi_{\hat{n},k} = \begin{bmatrix}
            \SY_{k-\hat{n},\hat{n}} \\ \SY_{k-\hat{n},\hat{n}}
        \end{bmatrix}
        =
        \begin{bmatrix}
            \SY_{k-\hat{n}+n,\hat{n}-n} \\ \SY_{k-\hat{n},n} \\ \SU_{k-\hat{n},\hat{n}}
        \end{bmatrix}
    \end{align*}
    For all $1 \le i \le \hat{n} - n$, note that $k-i = (k-\hat{n}) + n + (\hat{n}-n-i)$ and it follows from Proposition \ref{prop: IO j step update} that
    \begin{align*}
        & y_{ k-i}  
        = -\SF_{n,\hat{n}-n-i} \SY_{k-\hat{n},n} + \SG_{n,\hat{n}-n-i} \SU_{k-\hat{n},\hat{n}-i}
        \\
        & \ \ = -\SF_{n,\hat{n}-n-i} \SY_{k-\hat{n},n} 
        + \begin{bmatrix} 0_{p \times im} &  \SG_{n,\hat{n}-n-i} \end{bmatrix} \SU_{k-\hat{n},\hat{n}-i}.
    \end{align*}
    Hence,
    \begin{align*}
        \SY_{k-\hat{n}+n,\hat{n}-n} = \begin{bmatrix}
            M_{n,\hat{n}}^1 & M_{n,\hat{n}}^2
        \end{bmatrix}
        \begin{bmatrix}
            \SY_{k-\hat{n},n} \\ \SU_{k-\hat{n},\hat{n}}
        \end{bmatrix},
    \end{align*}
    and \eqref{eqn: phi = M*phi} follows from \eqref{eqn: M defn}.
\end{proof}

Proposition \ref{prop: equivalent theta} shows that input/output model \eqref{eqn: IO Model 2} with coefficients $\hat{\theta}_{\hat{n}}$ is equivalent to \eqref{eqn: IO Model} if and only if \eqref{eqn: equivalent theta} holds.
Next, Proposition \ref{prop: min norm equivalent theta} give an explicit formulation for the equivalent model of order $\hat{n}$ which minimizes the regularization term $\tr \left[(\hat{\theta}_{\hat{n}}-{\theta}_{\hat{n},0}) {P}_{\hat{n},0}^{-1} (\hat{\theta}_{\hat{n}}-{\theta}_{\hat{n},0})^\rmT \right]$ of the of cost function $J_{\hat{n},k}$, defined in \eqref{eqn: PCAC cost}.

\begin{prop}
\label{prop: equivalent theta}
    Input/Output models \eqref{eqn: IO Model} and \eqref{eqn: IO Model 2} are equivalent if and only if
    \begin{align}
    \label{eqn: equivalent theta}
        (\hat{\theta}_{\hat{n}} - {\theta}_{\hat{n},{\rm true}}) M_{n,\hat{n}} = 0_{p \times pn + m(\hat{n}+1)}.
    \end{align}
\end{prop}
\begin{proof}
    Note that an input/output model of order $\hat{n}$ with coefficients ${\theta}_{\hat{n},{\rm true}}$ is trivially equivalent to \eqref{eqn: IO Model}. Hence, by Theorem \ref{theo: IO Equivalence condition}, 
    ${\theta}_{\hat{n},{\rm true}} M_{n,\hat{n}} 
        =
        \begin{bmatrix}
            -{\SF}_{n,\hat{n}-n} & {\SG}_{n,\hat{n}-n}
        \end{bmatrix}$.
    Then, it also follows from Theorem \ref{theo: IO Equivalence condition} that \eqref{eqn: IO Model 2} is equivalent to \eqref{eqn: IO Model} if and only of $\hat{\theta}_{\hat{n}} M_{n,\hat{n}} = \begin{bmatrix}
            -{\SF}_{n,\hat{n}-n} & {\SG}_{n,\hat{n}-n}
        \end{bmatrix} = {\theta}_{\hat{n},{\rm true}} M_{n,\hat{n}} $.
\end{proof}

\begin{prop}
\label{prop: min norm equivalent theta}
    The constrained optimization problem
    \begin{align}
        \min_{\hat{\theta}_{\hat{n}} \in \BBR^{p \times \hat{n}(m+p) + m}} & \tr \left[ (\hat{\theta}_{\hat{n}} - {\theta}_{\hat{n},0}) P_{\hat{n},0}^{-1} (\hat{\theta}_{\hat{n}} - {\theta}_{\hat{n},0})^\rmT \right],
        \label{eqn: constrained optimization min equivalent model}
        \\
        \textnormal{such that } & (\hat{\theta}_{\hat{n}} - {\theta}_{\hat{n},{\rm true}}) M_{n,\hat{n}} = 0_{p \times pn + m(\hat{n}+1)}, \nonumber
    \end{align}
    has the unique solution
    \begin{align}
    \label{eqn: theta star}
        \theta^*_{\hat{n}} \triangleq {\theta}_{\hat{n},0} 
        + ({\theta}_{\hat{n},{\rm true}} - {\theta}_{\hat{n},0}) H_{\hat{n}},
    \end{align}
    where hat matrix $H_{\hat{n}} \in \BBR^{p\hat{n} + m(\hat{n}+1) \times p\hat{n} + m(\hat{n}+1)}$ is defined
    \begin{align}
        H_{\hat{n}} \triangleq M_{n,\hat{n}} (M_{n,\hat{n}}^\rmT P_{\hat{n},0} M_{n,\hat{n}})^{-1} M_{n,\hat{n}}^\rmT P_{\hat{n},0}.
    \end{align}
\end{prop}
\begin{proof}
    Note that 
    \begin{align*}
        \tr \big[ (\hat{\theta}_{\hat{n}} - {\theta}_{\hat{n},0}) & P_{\hat{n},0}^{-1} (\hat{\theta}_{\hat{n}} - {\theta}_{\hat{n},0})^\rmT \big] 
        \\
        &=
        \sum_{i=1}^p (\hat{\theta}_{\hat{n},i} - {\theta}_{\hat{n},0,i}) P_{\hat{n},0}^{-1} (\hat{\theta}_{\hat{n},i} - {\theta}_{\hat{n},0,i})^\rmT,
    \end{align*}
    where, for all $1 \le i \le p$, $\hat{\theta}_{\hat{n},i} \in \BBR^{1 \times \hat{n}(m+p) + m}$ and ${\theta}_{\hat{n},0,i} \in \BBR^{1 \times \hat{n}(m+p) + m}$ are the $i^{\rm th}$ row of $\hat{\theta}_{\hat{n}}$ and ${\theta}_{\hat{n},0}$, respectively.
    It then follows that \eqref{eqn: constrained optimization min equivalent model} can be written as $p$ separate optimization problems given by, for all $1 \le i \le p$,
    \begin{align}
    \label{eqn: ith row constrained opt}
        \min_{\hat{\theta}_{\hat{n},i} \in \BBR^{1 \times \hat{n}(m+p) + m}} & (\hat{\theta}_{\hat{n},i} - {\theta}_{\hat{n},0,i}) P_{\hat{n},0}^{-1} (\hat{\theta}_{\hat{n},i} - {\theta}_{\hat{n},0,i})^\rmT,
        \\
        \textnormal{such that } & M_{n,\hat{n}}^\rmT \hat{\theta}_{\hat{n},i}^\rmT  = M_{n,\hat{n}}^\rmT {\theta}_{\hat{n},{\rm true},i}^\rmT , \nonumber
    \end{align}
    where, for all $1 \le i \le p$, ${\theta}_{\hat{n},{\rm true},i} \in \BBR^{1 \times \hat{n}(m+p) + m}$ is the $i^{\rm th}$ row of ${\theta}_{\hat{n},{\rm true},i}$.
    It then follows from Lemma \ref{lem: quadratic minimization with affine constraints} that, for all $1 \le i \le p$, \eqref{eqn: ith row constrained opt} has a unique solution ${\theta}_{\hat{n},i}^*$, given by
    \begin{align}
        M_{n,\hat{n}}^\rmT {\theta_{\hat{n},i}^*}^\rmT &= M_{n,\hat{n}}^\rmT \theta_{\hat{n},{\rm true},i}^\rmT,
        \label{eqn: quad opt temp 1.1}
        \\
        P_{\hat{n},0}^{-1}{\theta_{\hat{n},i}^*}^\rmT + M_{n,\hat{n}} {\nu_i^*}^\rmT &= P_{\hat{n},0}^{-1} \theta_{\hat{n},0,i}^\rmT.
        \label{eqn: quad opt temp 1.2}
    \end{align}
    %
    It follows from \eqref{eqn: quad opt temp 1.2} that
    \begin{align}
    \label{eqn: quad opt temp 2}
        {\theta_{\hat{n},i}^*}^\rmT = \theta_{\hat{n},0,i}^\rmT - P_{\hat{n},0} M_{n,\hat{n}} {\nu_i^*}^\rmT.
    \end{align}
    Substituting \eqref{eqn: quad opt temp 2} into \eqref{eqn: quad opt temp 1.1}, it follows that
    \begin{align}
         M_{n,\hat{n}}^\rmT ( \theta_{\hat{n},0,i}^\rmT - P_{\hat{n},0} M_{n,\hat{n}} {\nu_i^*}^\rmT) = M_{n,\hat{n}}^\rmT \theta_{\hat{n},{\rm true},i}^\rmT.
    \end{align}
    Hence, ${\nu_i^*}^\rmT$ is given as
    \begin{align}
    \label{eqn: quad opt temp 3}
         {\nu_i^*}^\rmT =  ( M_{n,\hat{n}}^\rmT P_{\hat{n},0} M_{n,\hat{n}})^{-1} M_{n,\hat{n}}^\rmT ( \theta_{\hat{n},0,i}^\rmT - \theta_{\hat{n},{\rm true},i}^\rmT).
    \end{align}
    Substituting \eqref{eqn: quad opt temp 3} into \eqref{eqn: quad opt temp 2}, it follows that
    ${\theta_{\hat{n},i}^*}^\rmT = \theta_{\hat{n},0,i}^\rmT - P_{\hat{n},0} M_{n,\hat{n}} ( M_{n,\hat{n}}^\rmT P_{\hat{n},0} M_{n,\hat{n}})^{-1} M_{n,\hat{n}}^\rmT ( \theta_{\hat{n},0,i}^\rmT - \theta_{\hat{n},{\rm true},i}^\rmT).$
    Finally, taking the transpose and noting that ${\theta_{\hat{n},i}^*}$ is the $i^{\rm th}$ row of $\theta_{\hat{n}}^*$ yields \eqref{eqn: theta star}.
\end{proof}

As mentioned previously, the input/output model \eqref{eqn: IO Model 2} with coefficients given by $\theta_{\hat{n}}^*$ minimizes the regularization term of $J_{\hat{n},k}$ over the set of input/output models of order $\hat{n}$ equivalent to \eqref{eqn: IO Model}.
Theorem \ref{theo: higher-order model convergence} show that if $(\phi_{n,\hat{n},k}^\rmT)_{k=0}^\infty$ is weakly persistently exciting, then $\theta_{\hat{n},k}$ converges to $\theta_{\hat{n}}^*$.
Moreover, if $(\phi_{n,\hat{n},k}^\rmT)_{k=0}^\infty$ is persistently exciting, then the convergence of $\theta_{\hat{n},k} - \theta_{\hat{n}}^*$ to zero is asymptotically proportional to $\nicefrac{1}{k}$.

\begin{theo}
\label{theo: higher-order model convergence}
    Consider the assumptions and notation of Proposition \ref{prop: PCAC RLS}.
    If $(\phi_{n,\hat{n},k}^\rmT)_{k=0}^\infty$ is weakly persistently exciting, then 
    \begin{align}
    \label{eqn: higher-order model convergence WPE}
        \lim_{k \rightarrow \infty} \theta_{\hat{n},k} = \theta_{\hat{n}}^*,
    \end{align}
    where $\theta_{\hat{n}}^*$ is defined in \eqref{eqn: theta star}. If, additionally, $(\phi_{n,\hat{n},k}^\rmT)_{k=0}^\infty$ is persistently exciting, then
    \begin{align}
    \label{eqn: higher-order model convergence PE}
         & \lim_{k \rightarrow \infty} k(\theta_{\hat{n},k} - \theta_{\hat{n}}^*) \nonumber
        \\
        & \quad = (\theta_{\hat{n},0} - \theta_{\hat{n},{\rm true}}) M_{n,\hat{n}} W_{n,\hat{n}}^{-1} C_{n,\hat{n}}^{-1} W_{n,\hat{n}}^{-1}  M_{n,\hat{n}}^\rmT P_{\hat{n},0}, 
    \end{align}
    where $C_{n,\hat{n}} \triangleq \lim_{k \rightarrow \infty} \frac{1}{k} \sum_{k=0}^{k-1} \phi_{n,\hat{n},i} \phi_{n,\hat{n},i}^\rmT$ and $W_{n,\hat{n}} \triangleq M_{n,\hat{n}}^\rmT P_{\hat{n},0} M_{n,\hat{n}}$.
\end{theo}
\begin{proof}
     For brevity, denote $P_0 \triangleq P_{\hat{n},0}$, $\theta^* \triangleq \theta_{\hat{n}}^*$, $\theta_0 \triangleq \theta_{\hat{n},0}$, $\theta \triangleq \theta_{\hat{n},{\rm true}}$, $M \triangleq M_{n,\hat{n}}$, $H = H_{\hat{n}}$, $W \triangleq W_{n,\hat{n}}$, and, for all $k \ge 0$, $\theta_k \triangleq \theta_{\hat{n},k}$, $\Phi_k \triangleq {\Phi}_{\hat{n},k}$ and $\bar{\Phi}_{k} \triangleq \begin{bmatrix}
        \phi_{n,\hat{n},k-1} & \cdots & \phi_{n,\hat{n},0}
    \end{bmatrix}.$
    It follows from \eqref{eqn: batch IO ID}, \eqref{eqn: yk = thetatrue nhat phi}, and Lemma \ref{lem: phi = M*phi} that, for all $k \ge 0$,
    \begin{align*}
        \theta_{k+1} &= (\theta \Phi_k \Phi_k^\rmT + \theta_0 P_0^{-1} ) (\Phi_k \Phi_k^\rmT + P_0^{-1})^{-1}
        \\
        &= \theta + (\theta_0 - \theta) P_0^{-1} (\Phi_k \Phi_k^\rmT + P_0^{-1})^{-1}
        \\
        &= \theta + (\theta_0 - \theta) P_0^{-1} (M \bar{\Phi}_k \bar{\Phi}_k^\rmT M^\rmT + P_0^{-1})^{-1}.
    \end{align*}
    Subtracting both sides by $\theta^*$ and substituting \eqref{eqn: theta star} yields
    \begin{align}
    \label{eqn: higher-order model convergence temp 2}
        &\theta_{k+1} - \theta^* 
        \\
        & \quad = (\theta_0 - \theta) \big[-I +  P_0^{-1} (M \bar{\Phi}_k \bar{\Phi}_k^\rmT M^\rmT + P_0^{-1})^{-1} + H \big]. \nonumber
    \end{align}
    Since $(\phi_{n,\hat{n},k}^\rmT)_{k=0}^\infty$ is weakly persistently exciting, there exists $N$ such that, for all $k \ge N$, $\bar{\Phi}_k \bar{\Phi}_k^\rmT$ is nonsingular. 
    Then, it follows from Lemma \ref{lem: matrix inversion lemma} that, for all $k \ge N$,
    \begin{align*}
        (M \bar{\Phi}_k \bar{\Phi}_k^\rmT M^\rmT & + P^{-1})^{-1} 
        \\
        & = P_0 - P_0 M [(\bar{\Phi}_k \bar{\Phi}_k^\rmT)^{-1} + W]^{-1} M^\rmT P_0,
    \end{align*}
    and hence
    \begin{align}
        -I +  P_0^{-1} (M \bar{\Phi}_k & \bar{\Phi}_k^\rmT M^\rmT + P_0^{-1})^{-1} \nonumber
        \\
        & = - M [(\bar{\Phi}_k \bar{\Phi}_k^\rmT)^{-1} + W ]^{-1} M^\rmT P_0.
        \label{eqn: higher-order model convergence temp 4}
    \end{align}
    Substituting \eqref{eqn: higher-order model convergence temp 4} and $H = M(M^\rmT P_0 M)^{-1} M^\rmT P_0$ into \eqref{eqn: higher-order model convergence temp 2} implies that, for all $k \ge N$,
    \begin{align*}
        \theta_{k+1} &  - \theta^* 
        \\
       & = (\theta_0 - \theta) M \left[-[(\bar{\Phi}_k \bar{\Phi}_k^\rmT)^{-1} + W ]^{-1} + W^{-1} \right] M^\rmT P_0.
    \end{align*}

    Again applying Lemma \ref{lem: matrix inversion lemma}, it follows that 
    \begin{equation*}
        -[(\bar{\Phi}_k \bar{\Phi}_k^\rmT)^{-1} + W]^{-1}  + W^{-1} 
        =  W^{-1} [ \bar{\Phi}_k \bar{\Phi}_k^\rmT - W^{-1}]^{-1} W^{-1} 
    \end{equation*}
    which yields
    \begin{equation*}
        \theta_{k+1} - \theta^* = (\theta_0 - \theta) M W^{-1} [ \bar{\Phi}_k \bar{\Phi}_k^\rmT - W^{-1}]^{-1} W^{-1}  M^\rmT P_0.
    \end{equation*}
    Since $(\phi_{n,\hat{n},k}^\rmT)_{k=0}^\infty$ is weakly persistently exciting, it follows that 
    $\lim_{k \rightarrow \infty} \eigmin \left( \bar{\Phi}_k \bar{\Phi}_k^\rmT \right) = \infty$ and $\lim_{k \rightarrow \infty} \theta_{k+1} - \theta^* = 0$ and \eqref{eqn: higher-order model convergence WPE} follows.
    If additionally, $(\phi_{n,\hat{n},k}^\rmT)_{k=0}^\infty$ is persistently exciting, then
    \begin{align*}
        & \lim_{k \rightarrow \infty} k (\theta_{k+1} - \theta^*) 
        \\
        &= (\theta_0 - \theta) M W^{-1} [ \lim_{k \rightarrow \infty} \frac{1}{k}\bar{\Phi}_k \bar{\Phi}_k^\rmT - \frac{1}{k} W^{-1}]^{-1} W^{-1}  M^\rmT P_0.
    \end{align*}
    Since $\lim_{k \rightarrow \infty} \frac{1}{k}\bar{\Phi}_k \bar{\Phi}_k^\rmT - \frac{1}{k} W^{-1} = C_{n,\hat{n}}$, \eqref{eqn: higher-order model convergence PE} follows.
\end{proof}
Finally, since the input/output model \eqref{eqn: IO Model 2} with coefficients given by $\theta_{\hat{n}}^*$ is equivalent to \eqref{eqn: IO Model} and $\lim_{k \rightarrow \infty} {\theta}_{\hat{n},k} = \theta_{\hat{n}}^*$, an immediate corollary is that the residual error $z_{\hat{n},k}({\theta}_{\hat{n},k})$ converges to $0$, where $z_{\hat{n},k}$ is defined in \eqref{eqn: zk = y - theta phi}.
\begin{cor}
\label{cor: residual error approaches 0}
Consider the assumptions and notation of Proposition \ref{prop: PCAC RLS}.
If $(\phi_{n,\hat{n},k}^\rmT)_{k=0}^\infty$ is weakly persistently exciting, then 
    \begin{align}
        \lim_{k \rightarrow \infty} z_{\hat{n},k}({\theta}_{\hat{n},k}) = 0_{p \times 1}.
    \end{align}
\end{cor}



\section{Conclusions}

This work shows that when using RLS to identify an input/output model of a higher-order than that of the true input/output system, the identified coefficients still converge to predictable values, given weakly persistently exciting data.
In particular, we obtain the very natural result that the higher-order identified model converges to the model equivalent to the true system that minimizes the regularization term of RLS.

One outstanding question in our analysis is under what conditions the persistent excitation condition of Theorem \ref{theo: higher-order model convergence} is met.
This itself is a deep question with sufficient conditions developed in \cite{willems2005note} and a relaxation of those conditions presented in \cite{van2020willems}.
These conditions require that the input be persistently exciting of a sufficiently high order, see \cite{willems2005note} for further details.
One method to guarantee the input is persistently exciting of arbitrarily high order is to introduce random inputs in, for example, an $\varepsilon$-greedy control strategy \cite{watkins1989learning}.
Designing a suitable input that is persistently exciting of sufficiently high order relates to the classical dilemma of exploration and exploitation \cite{sutton2018reinforcement} and is beyond the scope of this work.

Another issue not addressed explicitly is the case where the order of the identified input/output model is lower than that of the true input/output system. 
One immediate result however, is that if the true system is irreducible, then the identified model cannot converge to an equivalent model, as no lower-order equivalent model exists.
As a result, the residual error in the RLS cost will not converge to zero, unlike in the higher-order case (see Corollary \ref{cor: residual error approaches 0}).

An important application of this work is the selection of a model order for online identification. 
Many criterion rules for model-order selection balance the trade-off of maximizing goodness of fit of the model (minimizing residual error in this case) and the simplicity of the model (minimizing model order in this case) \cite{stoica2004model}.
This work shows how model identification of different orders perform under idealized conditions, which may help inform model order selection in more realistic conditions.



\bibliographystyle{IEEEtran}
\bibliography{refs}

\appendix
\begin{lema}{A.1}
\label{lem: matrix inversion lemma}
Let $A \in \BBR^{n \times n}$, $U \in \BBR^{n \times p}$, $C \in \BBR^{p \times p}$, $V \in \BBR^{p \times n}$. Assume $A$, $C$, and $A+UCV$ are nonsingular. Then, $(A+UCV)^{-1} = A^{-1} - A^{-1}U(C^{-1} + VA^{-1} U)^{-1} V A^{-1}$.
\end{lema}



\begin{lema}{A.2}
\label{lem: quadratic minimization with affine constraints}
Consider the quadratic minimization problem with affine constraints
\begin{equation}
\label{eqn: quad prog}
\begin{aligned}
    \min_{x \in \BBR^n} &\frac{1}{2} x^\rmT P x + q^\rmT x + r, \\
    \textnormal{such that } &A x = b,
\end{aligned}   
\end{equation}

where $P \in \BBR^{n \times n}$ is positive definite and $A\in \BBR^{p \times n}$. Then, \eqref{eqn: quad prog} has a unique solution $(x^*,\nu^*)$ given by the system
\begin{align}
    A x^* = b, \quad Px^*  + A^\rmT \nu^* = -q.
\end{align}
%
\end{lema}
\begin{proof}
    See section 10.1.1 of \cite{boyd2004convex}.
\end{proof}

\end{document}